\documentclass[letterpaper, 10 pt, conference]{ieeeconf}  

\IEEEoverridecommandlockouts                              

\title{\LARGE \bf
Risk-Budgeted Control Framework for Improved Performance and Safety in Autonomous Vehicles
}

\author{ Pei Yu Chang, Vishnu Renganathan, and Qadeer Ahmed
\thanks{This research was supported by the CARMEN+ University Transportation Center, sponsored by the U.S. Department of Transportation under Grant No. 69A3552348327. The views presented are those of the authors and do not necessarily represent the official views of the U.S. Department of Transportation.}
\thanks{Pei Yu Chang, Qadeer Ahmed are
with the Department of Mechanical and Aerospace Engineering, The Ohio
State University, Columbus, OH 43212 USA. Email:
\{chang.2314, ahmed.358\}@osu.edu}
\thanks{Vishnu Renganathan is
with the Department of Electrical and Computer Engineering, The Ohio State
University, Columbus, OH 43212 USA. Email:
renganathan.5@osu.edu}}

\usepackage{amsmath,amssymb,amsfonts,mathtools}
\usepackage{algorithm}
\usepackage{algpseudocode}
\usepackage{graphicx}
\usepackage{booktabs}
\usepackage{hyperref}
\usepackage{bm}
\usepackage{algorithm}
\usepackage{algpseudocode}
\usepackage{cite}
\usepackage{stfloats}
\usepackage[table]{xcolor} 
\usepackage{makecell} 

\newtheorem{theorem}{Theorem}
\newtheorem{lemma}{Lemma}

\usepackage{graphicx}
\graphicspath{{figure/}}
\begin{document}

\maketitle
\thispagestyle{empty}
\pagestyle{empty}

\begin{abstract}

This paper presents a hybrid control framework with a risk-budgeted monitor for safety-certified autonomous driving. A sliding-window monitor tracks insufficient barrier residuals and triggers switching from a relaxed control barrier function (R-CBF) to a more conservative conditional value-at-risk CBF (CVaR-CBF) when the safety margin deteriorates. Two real-time triggers are considered: feasibility-triggered (FT), which activates CVaR-CBF when the R-CBF problem is reported infeasible, and quality-triggered (QT), which switches when the residual falls below a prescribed safety margin.

The framework is evaluated with model predictive control (MPC) under vehicle localization noise and obstacle position uncertainty across multiple AV-pedestrian interaction scenarios with 1,500 Monte Carlo runs. In the most challenging case with 5 m pedestrian detection uncertainty, the proposed method achieves a 94--96\% collision-free success rate over 300 trials while maintaining the lowest mean cross-track error (CTE = 3.2--3.6 m), indicating faster trajectory recovery after obstacle avoidance and a favorable balance between safety and performance.

\end{abstract}


\section{Introduction}

Safety under uncertainty is a fundamental requirement in autonomous vehicle (AV) control, and the challenge is amplified when safety filters are implemented as real-time optimization routines operating at high frequencies (e.g., $20$--$50$\,Hz) \cite{Renganathan2025Experimental}. Under strict computational budgets, the online solver may return (i) a feasible (near-)optimal solution, (ii) a feasible but non-optimal iterate due to early termination or imperfect convergence, or (iii) a solver-reported infeasibility outcome. In such real-time settings, reliably quantifying the degree of suboptimality---or determining whether a returned iterate remains sufficiently safe---on a step-by-step basis is generally nontrivial. Accordingly, we study safety under \emph{non-ideal} solver outcomes, encompassing both suboptimal feasible iterates and solver-reported infeasibility.

Control Barrier Functions (CBFs) formulated as quadratic programs (QPs) provide a principled mechanism for enforcing forward invariance, guaranteeing safety when the CBF-QP remains feasible and is solved with sufficient optimality \cite{ames2014cbfacc, ames2017cbfqpsafety, mrdjan2018inputdelay}. To improve deployability under input constraints and modeling complexity, prior work has considered relaxed CBFs with slack penalties \cite{xiao2022adaptivecbf}, input-aware formulations \cite{rabiee2024closedformcontrolsafetyinput, Agrawal_2021}, higher-order CBFs \cite{xiao2022hocbf}, and learning-based tuning of controller or safety parameters \cite{sabouni2025reinforcementlearningbasedrecedinghorizon, xiao2023learningfeasibilityconstraintscontrol}. While these methods mitigate conservatism and feasibility issues in deterministic settings, environmental uncertainty can still erode safety margins or induce overly conservative behavior.
\begin{figure}
      \centering
      \includegraphics[width=\columnwidth]{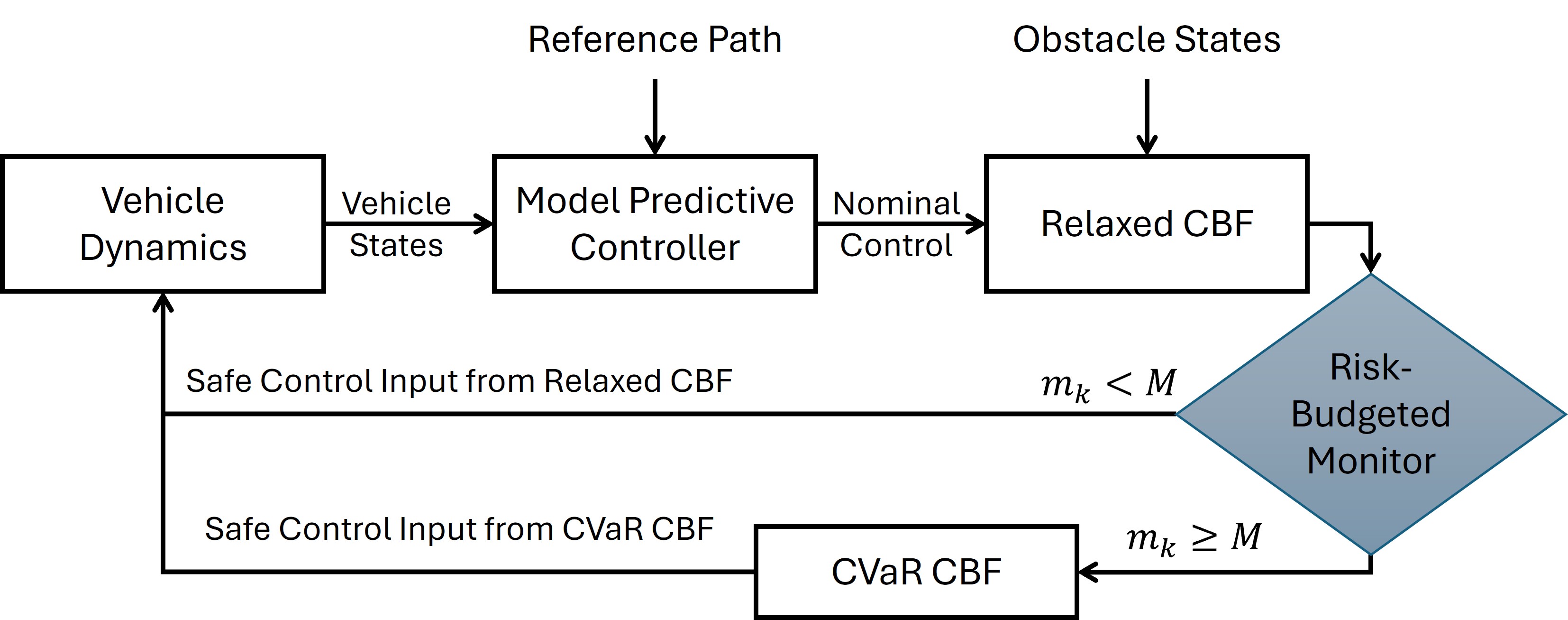}
      \caption{Proposed framework with MPC nominal control and a risk–budget monitor for safety–driven controller switching.}
      \label{framework}
\end{figure}

Under uncertainty, stochastic and CVaR-based CBF formulations explicitly address distributional tail risk \cite{clark2021cbfstochastic, ahmadi2022riskcbf, wang2025safenavigationuncertaincrowded}. These approaches regulate confidence levels to trade off feasibility and conservatism, sometimes adapting risk online to preserve probabilistic safety guarantees. However, chance-constrained and stochastic safety filters may still be overly conservative in practice and often incur higher computational cost. Although CVaR--CBF recovers the deterministic CBF case in the limit, enforcing risk-aware optimization at every step is often unnecessary for high-rate real-time deployment.


Related efforts include learning-based feasibility recovery and probabilistic barrier methods that adapt constraints or risk parameters online (e.g., \cite{ahmadi2022riskcbf}), as well as studies that provide safety guarantees under solver-induced non-optimality (e.g., early termination or imperfect convergence) \cite{Allibhoy2021, Wang_2023}. In contrast to per-step feasibility approaches, we certify safety over a finite horizon using window-level risk budgeting under occasional non-ideal solver outcomes.

This paper adopts a practical perspective: non-ideal solver outcomes do not necessarily imply immediate safety loss, but repeated occurrences can gradually erode the safety margin. Instead of evaluating safety at each step, we monitor the evolution of the safety margin over a sliding window. The proposed risk-budgeted monitor evaluates safety quality over a finite horizon, allowing occasional undesired control actions due to computational or input limitations while still certifying windowed safety.


Since monitoring alone cannot enforce safety, the monitor output is used as a switching signal in a hybrid architecture that localizes conservatism to periods of accumulated risk. The nominal mode employs a lightweight relaxed CBF-QP to preserve tracking performance, while a conservative CVaR--CBF mode is activated when the monitored safety quality degrades. The overall architecture is illustrated in Fig.~\ref{framework}.

The main contributions of this work are:
\begin{enumerate}
    \item A finite-horizon risk-budget formulation that certifies safety under transient non-optimal solutions.
    \item A monitor-driven switching mechanism between performance-oriented CBF-QP and conservative CVaR--CBF filters.
    \item Theoretical analysis and simulations demonstrating balanced safety and performance under uncertainty.
\end{enumerate}

\section{Preliminaries}

Consider a nonlinear control–affine system
\begin{equation} \label{control_affine}
\dot{\bm{x}} = f(\bm{x}) + g(\bm{x})\bm{u},
\end{equation}
where $\bm{x} \in \mathcal{D} \subset \mathbb{R}^n$ and $\bm{u} \in \mathcal{U} \subset \mathbb{R}^m$ denote the state and input, respectively.
Let the safe set be defined as
\[
\mathcal S=\{\bm{x}:h(\bm{x})\ge 0\},
\]
where $h:\mathcal D\rightarrow\mathbb R$ is a continuously differentiable function.
Safety is guaranteed if the control input satisfies the CBF condition \cite{ames2019controlbarrierfunctions}
\begin{equation}
L_f h(\bm{x})+L_g h(\bm{x})\bm{u}+\kappa h(\bm{x})\ge0 .
\label{eq:zcbf_cts}
\end{equation}

Define the safety residual
\begin{equation}
r(\bm{x},\bm{u})
:=
L_f h(\bm{x})+L_g h(\bm{x})\bm{u}+\kappa h(\bm{x}).
\label{eq:residuals}
\end{equation}

For real–time implementation, the controller is applied with zero–order hold,
$\bm u(t)=u_k$ for $t\in[t_k,t_{k+1})$.
Under this sampling scheme, the barrier dynamics satisfy
\begin{equation}
\dot h(t)\ge -\kappa h(t)+r_k,
\label{eq:cbf_timevariant}
\end{equation}
where $r_k=\inf_{t\in[t_k,t_{k+1})} r(\bm x,u_k)$.

Solving \eqref{eq:cbf_timevariant} yields the one–step comparison inequality
\begin{equation}
h_{k+1}\ge \mu h_k + c r_k,
\label{eq:disc_compare}
\end{equation}
where $\mu=e^{-\kappa T_s}\in(0,1)$ and $c=(1-\mu)/\kappa$.

\subsection{Relaxed CBF (R-CBF)}
The theory thus far yields a safe control input $u_k$ while a nominal controller input $u_k^{\text{nom}}$ is provided by the MPC in our implementation at time $k$. Thus, an optimal safe controller is implemented in a QP formulation using a CBF-based safety filter with a nonnegative slack. The slack addresses infeasibility due to actuator limitations as shown in ~\cite{ames2019controlbarrierfunctions}:
\begin{equation} \label{eq:relaxed_cbf}
\begin{aligned}
\min_{u_k \in \mathcal{U},\; \nu_k \ge 0}\quad
  & \tfrac12 \|u_k - u_k^{\text{nom}}\|^2 \;+\; \rho_\nu\, \nu_k^2 \\
\text{s.t.}\quad
  & L_f h(x_k) + L_g h(x_k)\,u_k + \kappa\, h(x_k) \;\ge\; -\,\nu_k .
\end{aligned}
\end{equation}
The slack $\nu_k>0$ quantifies the instantaneous relaxation of the barrier constraint and is penalized by $\rho_\nu>0$. A smaller penalty admits larger (more aggressive and performance-seeking) relaxations, while a larger penalty on $\nu_k$ yields more conservative (safety-biased) behavior.
In Section \ref{risk_monitor}, we demonstrate how to utilize $\nu_k$ and the residual margins $r_k$ as supervisory signals that guide a balanced control framework between performance and safety.

\subsection{Stochastic Safe Sets}

Given uncertainty in obstacle and vehicle states, safety is evaluated in a stochastic setting. Let $\bm{x}^{o,p}\in\mathcal O$ and $\bm{x}^{v,q}\in\mathcal D$ denote $P$ obstacle samples and $Q$ vehicle samples, respectively, yielding $S=PQ$ vehicle--obstacle pairs indexed by $i$. The CBF for each pair at time $k$ is
\begin{subequations}
\begin{align}
h_k^i &= \left\|{\bm{x}}_k^{v,q} - {\bm{x}}_k^{o,p}\right\|_2 - D_s, \\
h_{k+1}^{i} &\geq \mu h_{k}^{i} + c r_{k}^{i}, \\
r_k^{i}(u_k) &= L_f h_k^i + L_g h_k^i u_k + \kappa h_k^i,
\end{align}
\end{subequations}
where $i=1,\dots,S$ and $D_s$ is the minimum safe distance.

To align the goal of maximizing safety with a standard minimization problem, we define a loss function $Z_k^i(u_k)$ as the negative of the safety residual:
\begin{equation}
Z_{k+1}^i(u_k):=-r_k^{i}(u_k).
\label{eq:loss_function}
\end{equation}
This definition creates an intuitive relationship: a highly safe outcome (large positive $r_k^i$) corresponds to a very low loss (negative $Z_k^i$), while an unsafe outcome (negative $r_k^i$) results in a high, positive loss. This transformation turns the safety problem into a convex, sample-average tail control problem.
 
\subsection{Risk-Aware Safety Critical Control Formulation}

\subsubsection{Probability Constraints}


Due to uncertainty in obstacle and vehicle positions, safety is imposed through the chance constraint
\begin{equation} \label{eq:pro_Constraints}
    \mathbb{P} \left( Z_{k+1}\leq 0 \right) \geq \epsilon,
\end{equation}
where $\epsilon\in(0,1)$ is the confidence level.
Let $Z_{k+1}$ denote the random loss at the next step, with sampled realizations $\{Z_{k+1}^i\}_{i=1}^S$.
The chance constraint in \eqref{eq:pro_Constraints} can be expressed using the Value-at-Risk (VaR) measure.
\begin{equation}\label{eq:VaR_Constraint}
\text{VaR}_{\epsilon}(Z_{k+1}) := \inf_{\gamma \in \mathbb{R}} \left\{ \gamma \;\middle|\; \mathbb{P}(Z_{k+1} \leq \gamma) \geq \epsilon \right\},
\end{equation}
where $\gamma$ represents the candidate loss threshold being evaluated. 
Thus, the chance constraint is equivalent to requiring that the $\epsilon$-quantile of the loss distribution does not exceed zero.
Hence, the probabilistic safety requirement is
\begin{equation}\label{eq:VaR_final}
\text{VaR}_{\epsilon}(Z_{k+1}) \leq 0.
\end{equation}

While VaR specifies the $\epsilon$-quantile of the loss distribution, it does not account for the severity of losses in the $(1-\epsilon)$ tail. To capture this tail risk, we adopt Conditional Value-at-Risk (CVaR).
The Conditional Value-at-Risk (CVaR) at confidence level $\epsilon$ is the expected loss in the $(1-\epsilon)$ tail \cite{sergey2008cvar}:
\begin{equation}\label{eq:cvar_var}
\text{CVaR}_{\epsilon}(Z_{k+1})
:=
\mathbb{E}\!\left[ Z_{k+1} \mid Z_{k+1} \ge \text{VaR}_{\epsilon}(Z_{k+1}) \right].
\end{equation}

Using the Rockafellar--Uryasev representation \cite{rockafellar2000optimization},
\begin{equation}\label{eq:cvar_criteria}
\text{CVaR}_{\epsilon}(Z_{k+1})
=
\inf_{\gamma \in \mathbb{R}}
\left\{
\gamma + \frac{1}{1-\epsilon}\mathbb{E}\!\left[(Z_{k+1}-\gamma)_+\right]
\right\},
\end{equation}
where $(x)_+ := \max(0,x)$.
In practice, enforcing zero risk at every step may render the safety filter infeasible under uncertainty. We therefore introduce a nonnegative slack variable $\nu_k$ to quantify the admitted risk at time $k$, and define the resulting filter as the Relaxed CVaR--CBF. Given a nominal command $u_k^{\mathrm{nom}}$, the control input is obtained by solving
\begin{equation}\label{eq:cost_function}
\begin{aligned}
\min_{u_k \in \mathcal{U}, \ \nu_k \in \mathbb{R}} \quad & \frac{1}{2}\|u_k-u_k^{\mathrm{nom}}\|^2 + \rho_{\nu}\nu_k^2, \\
\text{s.t.}\quad 
& \mathrm{CVaR}_{\epsilon}\!\left(Z_{k+1}\right) \le \nu_k,\\
& r_k^i(u_k) \ge -\,\nu_k,\quad i=1,\dots,S,\\
& 0 \le \nu_k \le \overline{\nu},
\end{aligned}
\end{equation}
where $\rho_\nu>0$ penalizes admitted risk. The constraint $\mathrm{CVaR}_{\epsilon}(Z_{k+1}) \le \nu_k$ bounds the tail risk, while $r_k^i(u_k)\ge -\nu_k$ enforces a uniform per-scenario lower bound on the residual required by the sliding-window analysis in Section~\ref{risk_monitor}. The upper bound $\overline{\nu}$ specifies the stepwise risk cap, whose admissible magnitude is determined by the window-level safety certificate in Section~\ref{window_paragraph}.

Under the sample-based Rockafellar--Uryasev epigraph reformulation, \eqref{eq:cost_function} becomes a convex quadratic program.

\section{Risk-Budget Monitor}\label{risk_monitor}

As the optimization in \eqref{eq:cost_function} allows non-zero risk through relaxation, a trade-off can be achieved between performance and safety. However, under real-time computational constraints, the solver may occasionally return a non-optimal solution, such as a suboptimal convergence or a temporarily elevated risk level. Importantly, a single non-optimal step or transient increase in risk does not necessarily drive the AV into an unsafe region, particularly when a safety margin is maintained.

To manage this trade-off systematically, this section formalizes the notion of a bad step and introduces a sliding-window risk monitor to evaluate safety quality over a finite horizon.

\subsection{Bad Step detection} 
We first define $\delta>0$ as the desired per-step safety margin. A step is considered \emph{good} if the residual satisfies $r_{k}\ge \delta$, which effectively deposits a safety allowance into the system. We further assume the safety filter enforces a uniform residual lower bound $r_{k}\ge -\overline{\nu}$ at every step \eqref{eq:cost_function}. A step is classified as a \emph{bad} step if it meets the following condition:

\begin{itemize} \label{bad_step_condition}

    \item \textbf{Condition: $r_{k} < \delta$} \\
    This condition flags a step as \textbf{bad} if the residual in R-CBF fails to meet this minimum safety margin $\delta$ required for a good step. This acts as a conservative check; a step is considered bad even if the overall statistical risk is zero ($\nu_k = 0$), ensuring that any single foreseeable negative outcome is accounted for.
\end{itemize}

The bad step indicator $b_k$ at time step $k$ is defined from the conditions as:
\begin{equation} \label{eq:bad_step_indicator}
b_k=\mathbf{1}\{r_{k}<\delta\}, \quad b_k \in \{0, 1\},
\end{equation}
i.e., $b_k=1$ denotes a bad step, otherwise the step is good.

\subsection{Sliding Window Risk-Budgeted Monitor}\label{window_paragraph}

In the previous section, we defined the conditions for classifying a control action as a bad step in \eqref{eq:bad_step_indicator}. At each time step, the system applies these criteria, effectively abstracting the complexity of the many individual samples into a deterministic good or bad classification. 

This classification is tracked by a risk monitor that operates over a sliding window of length $W$, ensuring that no more than $M$ bad steps, which withdraw from the system's safety allowance within the window. However, the allowable ratio of bad steps ($M/W$) is not arbitrary; it is fundamentally linked to the good-step safety margin, $\delta$, and the maximum permissible risk, $\overline{\nu}$. This section derives the precise mathematical relationship that governs these parameters, proving the safety of the monitor.

The analysis begins with the discrete single-step inequality from \eqref{eq:disc_compare}, which is recursively applied over a window of length $W$. This process evaluates the cumulative effect of a sequence of control actions, governed by the rule that no more than $M$ bad steps can occur.

By unrolling the inequality for $W$ steps, we obtain the following expression for the lower bound of $h_{k}$:
\begin{equation} \label{eq:unrolled_sum}
    h_{k+1} \ge \mu^{W}h_{k+1-W} + c \sum_{j=0}^{W-1} \mu^{W-1-j} \tilde{r}_{k-j}.
\end{equation}
Here, $\tilde{r}_{k-j}$ is a variable representing the worst-case lower bound of the residual $r_{k-j}$ used in the proof. Thus, $\tilde{r}_{k-j}$ takes a value from the set $\{\delta, -\overline{\nu}\}$. This is because a good step is guaranteed to have a residual of at least $\delta$ (in section \ref{bad_step_condition}), while a bad step has its residual lower-bounded by $-\overline{\nu}$. To guarantee safety ($h_{k+1} \ge 0$), we must find a lower bound for this expression. The core challenge lies in finding the worst-case arrangement of the $M$ bad steps in the summation. The next lemma establishes that the worst case occurs at the end of the step. Consequently, we can apply the rearrangement lemma to determine the parameters in this worst-case configuration.


\begin{lemma}[Rearrangement Lemma]
Fix any $k$ and define the window-indexed sequence $\{\tilde{r}_j\}_{j=0}^{W-1}$ by
$\tilde{r}_j := \tilde{r}_{k-j}$ for $j=0,\dots,W-1$.
Let the weights $w_j := \mu^{W-1-j}$ for $j=0, \dots, W-1$.
Since $\mu \in (0, 1)$, these weights are monotonically increasing:
$0 < w_0 < w_1 < \dots < w_{W-1}$.
Let $V$ be the set of sequences $\{\tilde{r}_j\}_{j=0}^{W-1}$ with the worst case $M$ entries equal to
$-\overline{\nu}$ and $W-M$ entries equal to $+\delta$.
The weighted sum $S(\tilde{r}) := \sum_{j=0}^{W-1} w_j \tilde{r}_{j}$ is minimized over $\tilde{r} \in V$
by assigning the $M$ negative values $(-\overline{\nu})$ to the largest weights and the $W-M$ positive values
$(+\delta)$ to the smallest weights.
\end{lemma}

\begin{proof}
Consider any pair of indices $i < j$, which implies $w_i \le w_j$. Let's assume an arrangement where $\tilde{r}_i = \delta$ and $\tilde{r}_j = -\overline{\nu}$. The contribution from this pair to the sum is $w_i \delta + w_j(-\overline{\nu})$.
By swapping the values: $\tilde{r}'_i = -\overline{\nu}$ and $\tilde{r}'_j = \delta$. The new contribution is $w_i(-\overline{\nu}) + w_j \delta$.
The difference between the original sum and swapped sum is:
\begin{align*}
(w_i \delta - w_j \overline{\nu}) - (-w_i \overline{\nu} + w_j \delta) &= w_i(\delta + \overline{\nu}) - w_j(\overline{\nu} + \delta) \\
&= (w_i - w_j)(\delta + \overline{\nu}).
\end{align*}
Since $w_i \le w_j$ and $\delta, \overline{\nu} \ge 0$, the difference $(w_i - w_j)(\delta + \overline{\nu})$ is less than or equal to 0. This means the sum is minimized when the smaller value ($-\overline{\nu}$) is paired with the larger weight ($w_j$). By repeatedly applying such pairwise swaps, we can prove that the minimum sum is achieved when all negative values are assigned to the largest weights.
\end{proof}

Applying the Rearrangement Lemma to the summation in Equation \eqref{eq:unrolled_sum}, we can establish the definitive worst-case lower bound for $h_{k+1}$. The minimum is achieved when the $M$ bad steps are assigned to the terms with the largest weights (corresponding to indices $j=W-M, \dots, W-1$) and the $W-M$ good steps are assigned to the terms with the smallest weights (indices $j=0, \dots, W-M-1$).

This splits the single summation into two parts. For notational clarity, we then re-index the sums. By substituting the summation index via the transformation $l = W-1-j$, the weights $\mu^{W-1-j}$ simply become $\mu^l$. This transformation maps the range for good steps ($j=0, \dots, W-M-1$) to $l=M, \dots, W-1$, and the range for bad steps ($j=W-M, \dots, W-1$) to $l=0, \dots, M-1$. This substitution yields the final lower bound:

\begin{equation}
h_{k+1} \ge \mu^{W}h_{k+1-W} + c\left[\delta\sum_{l=M}^{W-1}\mu^{l} - \overline{\nu}\sum_{l=0}^{M-1}\mu^{l}\right].
\label{eq:lower_bound_sum}
\end{equation}
This inequality represents the guaranteed system state under the worst possible sequence of actions and serves as the foundation for our main theorem.

\begin{theorem}[Window-Level Safety Certificate]\label{thm:window}
Assume the one-step comparison inequality \eqref{eq:disc_compare} holds with $0<\mu<1$ and $c>0$.
Suppose the safety filter enforces the uniform lower bound $r_t \ge -\overline{\nu}$ at every step $t$.
Fix any $k$ and consider the window of length $W$ ending at $k$, i.e., steps $t=k+1-W,\dots,k$.
Assume that at most $M$ steps in this window are \emph{bad} (i.e., satisfy $r_t<\delta$), and that the window starts from a safe state $h_{k+1-W}\ge 0$.
Then the terminal state at the end of the window is safe, i.e., $h_{k+1} \ge 0$, whenever
\begin{equation}\label{eq:window_certificate}
\mu^{M}(1-\mu^{W-M})\,\delta \ \ge\ (1-\mu^{M})\,\overline{\nu},
\end{equation}
where $\mu=e^{-\kappa T_s}$.
\end{theorem}

\begin{proof}
To ensure safety, we require that the lower bound of $h_{k+1}$ remains non-negative. Starting from the worst-case inequality established in \eqref{eq:lower_bound_sum} and conservatively dropping the $\mu^W h_{k+1-W} \ge 0$ term, the condition becomes:
\begin{equation} \label{eq:ineq_geom}
 c\left[\delta\sum_{l=M}^{W-1}\mu^{l} - \overline{\nu}\sum_{l=0}^{M-1}\mu^{l}\right] \ge 0 .
\end{equation}

Evaluating the geometric sums gives:
\begin{equation}\label{eq:geom}
\delta \frac{\mu^M(1-\mu^{W-M})}{1-\mu} \ge \overline{\nu} \frac{1-\mu^M}{1-\mu}. \end{equation}

Since $c>0$ and $0<\mu<1$ so that $(1-\mu)>0$, multiplying both sides of \eqref{eq:ineq_geom} by $(1-\mu)/c$ preserves the inequality direction and gives the certificate in \eqref{eq:window_certificate}.
\end{proof}

Rather than forbidding any instantaneous violation, the certificate caps how many bad steps and how severe they can be via $(W,M,\delta,\overline{\nu})$, thereby guaranteeing $h_{k+1}\ge 0$ at every window boundary. Once the $W, M, \delta$ has been designed based on the system requirement, we can get $\overline{\nu}$ from the equation below:

\begin{equation} \label{eq:nu_cap}
\overline{\nu} \le \frac{\mu^{M}(1-\mu^{W-M})}{1-\mu^{M}}\delta.
\end{equation}

The sliding-window certificate bounds the accumulation of bad steps and the resulting decay of the barrier value under the assumption that the system starts within the safe set, i.e., $h(x_0)\ge 0$, and that feasible control actions exist for the considered constraints.
Accordingly, the proposed monitoring mechanism is preventive. It regulates how often safety margin violations may occur so as to avoid leaving the safe set.
If the system state enters a region with $h(x)<0$ where no admissible input can drive it back to $h(x)\ge 0$ (due to actuation limits or environment geometry), then switching controllers cannot restore invariance and recovery is not guaranteed.

\section{Risk-Switch for Probabilistic Safety} \label{switch_mech}
Building on the risk monitor in Section~\ref{risk_monitor}, we formalize a sliding-window supervisor that certifies safety at the window scale and coordinates switching from the performance-seeking R-CBF filter to a conservative CVaR-CBF safety filter when risk accumulates. The cumulative number of bad steps $m_k$ in the most recent $W$-step window is updated by
\begin{equation}\label{eq:cumulate_bad}
  m_k \;=\; m_{k-1} \;+\; b_k \;-\; b_{k-W},
\end{equation}
with $m_k\in\{0,1,\dots,W\}$.

\subsection{Feasibility-Triggered CVaR-CBF (FT-C-CBF)}

To handle non-optimal outcomes in a unified way, we use the solver's feasibility status as a trigger.
At 20--50\,Hz, the solver often returns the same ``infeasible'' status for both true infeasibility and non-optimal termination, so we do not attempt to disambiguate online; instead, we treat it as a trigger and classify the step via the residual test $r_t<\delta$.

Let $\mathcal{F}_k$ denote the feasible set of the R-CBF filter at time $k$, and let
$a_k=\mathbf{1}\{(u_k,\nu_k)\in\mathcal{F}\}$ be its feasibility flag.
Under FT, feasibility serves as the primary guard, while the sliding-window counter enforces a budget on accumulated risk.
If the solver under the R-CBF constraints reports infeasible at time $k$ ($a_k=0$), the monitor updates $b_k$ and $m_k$ via \eqref{eq:cumulate_bad}.
When the solver reports infeasible at time $k$ ($a_k=0$), the monitor updates $b_k$ and $m_k$ via \eqref{eq:cumulate_bad}.
If an infeasible report occurs when the cumulative count has reached the budget (i.e., $a_k=0$ and $m_k\ge M$), the monitor switches from the R-CBF filter to the CVaR-CBF controller:
\[
\eta_k =
\begin{cases}
1, & \text{if } (a_k=0)\ \land\ (m_k \ge M),\\
0, & \text{otherwise} ,
\end{cases}
\]
where $\eta_k\in\{0,1\}$ is the controller mode ($0$: R-CBF, $1$: CVaR-CBF). The mode $\eta_k$ is implemented at time $k$.

\subsection{Quality-Triggered CVaR-CBF (QT-C-CBF)}

The quality trigger is motivated by the observation that optimality alone does not guarantee sufficient future safety margin.
Accordingly, at every step we check whether the residual satisfies the safety-margin requirement, and mark the step as bad whenever $r_k<\delta$.
Under QT, each step is vetted after solving the R-CBF constraints: the risk monitor computes $b_k$ and updates the sliding-window count $m_k$.
In contrast to the feasibility-triggered scheme, switching depends only on whether the window-level risk budget is exceeded:
\[
\eta_k =
\begin{cases}
1, & \text{if } m_k \ge M, \\
0, & \text{if } m_k < M,
\end{cases}
\]
Thus, even when the solver returns an optimal solution, a step flagged as bad by the monitor contributes to $m_k$; once the cumulative count reaches the budget $M$, the supervisor switches from the R-CBF filter to the CVaR controller.

\begin{figure*}
      \centering
       \setlength{\abovecaptionskip}{0pt}
      \includegraphics[width=\textwidth]{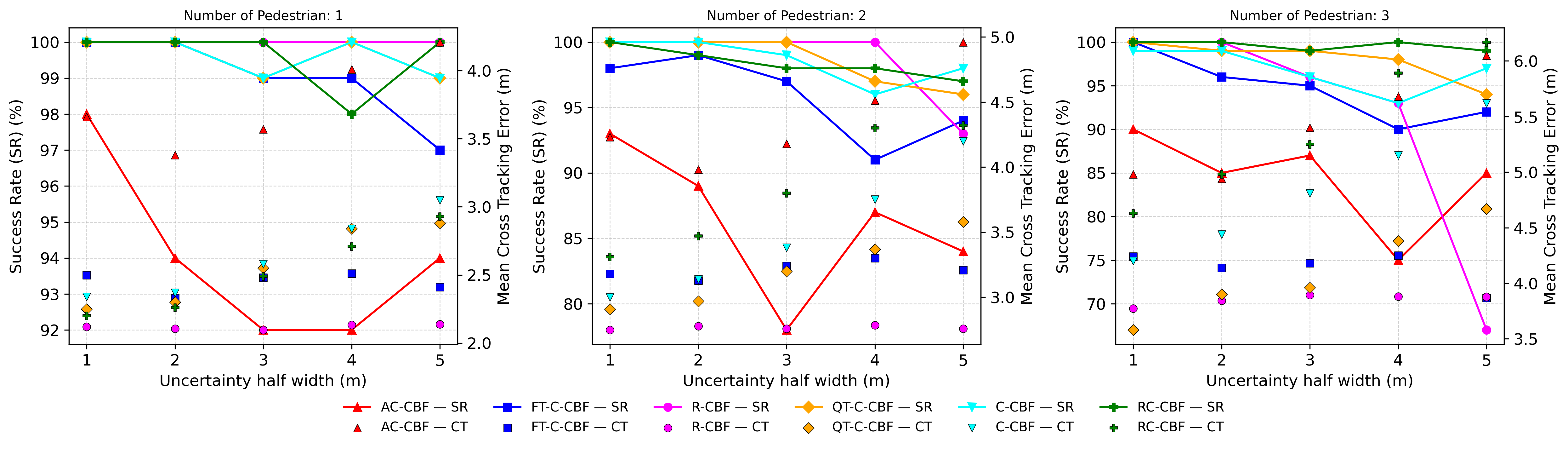}
      \caption{Success rate versus mean cross–track error (CTE) relative to the reference trajectory across 1,500 runs.}
      \label{sr_sensor}
\end{figure*}
\section{Simulations}
\subsection{Experimental Setup}
\subsubsection{Environment}


We evaluate a kinematic vehicle controlled by MPC following the experimental setup of \cite{Renganathan2025Experimental}. The reference trajectory is obtained from a GPS path recorded by the experimental AV in \cite{Renganathan2025Experimental}.

\subsubsection{Safety Framework Configuration}
The minimum safety distance is $D_s=3$ m and $\kappa=1$. The sliding-window monitor uses $W=5$, allowing at most $M=1$ bad step with a safety margin $\delta=1$ m. Based on these parameters, the risk cap computed from \eqref{eq:nu_cap} is $\overline{\nu}=3.8$.

For the CVaR constraint, the confidence level is $\epsilon=95\%$ with loss threshold $\gamma=0$, consistent with \eqref{eq:cvar_criteria}. Vehicle localization uncertainty is modeled as a Gaussian distribution with standard deviation $\sigma_v=0.1$. At each timestep, $Q=10$ samples are drawn around the measured position to represent possible true vehicle states.

\subsubsection{Obstacle Uncertainty Modeling}
To account for obstacle motion and perception uncertainty, the obstacle state is sampled at each timestep. Given the measured position $\bm{x}_{m,k}^{o}$, we generate $P=10$ samples $\bm{x}_k^{o,p}$ uniformly within a square bounding box centered at $\bm{x}_{m,k}^{o}$ with half width $\sigma_o \in \{1,2,3,4,5\}$ m. This models a perception system that reports a bounded box estimate of the obstacle location.
\begin{equation}
\begin{aligned}
\bm{x}_k^{o,p} &= \bm{x}_{m,k}^{o} + \Delta \bm{x}_k^{o,p}, \\
\Delta \bm{x}_k^{o,p} &\sim \mathcal{U}\!\left([-\sigma_o,\sigma_o]^2\right).
\end{aligned}
\end{equation}
\subsubsection{Performance Metrics}
We evaluate performance using five aggregate metrics, adapted in part from \cite{wang2025safenavigationuncertaincrowded}. Let $N_t$ denote the total number of test cases and $N_s$ the number of successful runs. Except for success rate, all metrics are averaged over successful runs only.

\begin{itemize}
    \item \text{Success Rate (SR):} A run is successful if the minimum vehicle--pedestrian distance exceeds $d_{\min}=2.8$ m throughout the trajectory. The success rate is
    $
    \mathrm{SR}=\frac{N_s}{N_t}.
    $
    
    \item \text{Minimum Distance to Pedestrian (MDP):} Let $d_{m,i}$ denote the minimum distance to the nearest pedestrian in successful run $i$. We report
    $
    \mathrm{MDP}=\frac{1}{N_s}\sum_{i=1}^{N_s} d_{m,i}.
    $
    
    A clearance of at least $5.8$ m satisfies the desired safety margin.

    \item \text{Infeasibility Rate (IR):} For successful run $i$, let $IR_i=n_i/n_t$, where $n_i$ is the number of infeasible steps and $n_t$ is the total number of optimization steps. We report
    $
    \mathrm{IR}=\frac{1}{N_s}\sum_{i=1}^{N_s} IR_i.
    $
    
    \item \text{Calculation Time (CT):} The mean per-step computation time over successful runs:
    $
    \mathrm{CT}=\frac{1}{N_s}\sum_{i=1}^{N_s} T_i.
    $
    
    \item \text{Cross-Track Error (CTE):} Let $E_i$ denote the average cross-track error in successful run $i$. We report
    $
    \mathrm{CTE}=\frac{1}{N_s}\sum_{i=1}^{N_s} E_i.
    $
\end{itemize}

\subsection{Comparative Analysis}

In the comparative evaluation, MPC serves as the nominal controller and is composed of the following safety filters:

\begin{enumerate}

\item {Relaxed CBF (R-CBF, baseline):}
We apply the relaxed CBF filter \eqref{eq:relaxed_cbf} in a plug-in manner using the raw noisy pedestrian measurements as the state input. This baseline represents a common naïve deployment of deterministic CBF filters under sensing noise, where estimation uncertainty is not modeled in the constraint. The purpose is to quantify the robustness gap caused by ignoring uncertainty.

\item {CVaR–CBF (hard tail constraint) (C-CBF):} The CVaR surrogate in \eqref{eq:cvar_criteria} and remove risk slack by fixing $\nu_k \equiv 0$ in \eqref{eq:cost_function}. This yields a strict tail-risk constraint, $E^+(u_k)\le 0$, with no per-step relaxation.

\item Adaptive CVaR--CBF (AC-CBF): This baseline adopts the risk-adaptive confidence-level idea similar to \cite{wang2025safenavigationuncertaincrowded}, but does not include the dynamic zone design or their CVaR formulation. Instead, we retain the formulation in \eqref{eq:cvar_criteria} and treat the confidence level as an adaptive parameter, selecting $\epsilon \in \{0.01,\,0.05,\,0.10,\,0.15,\,0.20,\,0.25,\,0.30,\,0.40,\,0.50\}$.

\item {Relaxed CVaR–CBF (RC-CBF):} We use the relaxed CVaR–CBF defined by \eqref{eq:cvar_criteria} and \eqref{eq:cost_function} throughout the trajectory, allowing $\nu_k \ge 0$ and no risk monitor.
\end{enumerate}

The results are evaluated with 1,500 Monte–Carlo runs under vehicle localization noise $\sigma_v=0.1$ and pedestrian position noise $\sigma_o\in\{1,2,3,4,5\}\,\mathrm{m}$. We consider three different pedestrian interaction scenarios: (i) one pedestrian performing a jaywalk across the AV’s path, (ii) two pedestrians, and (iii) three pedestrians, with 1,500 total runs. 
Figure~\ref{sr_sensor} organizes results by pedestrian number. The figure reports the detection bounding-box uncertainty; the left y-axis shows success rate (\%), and the right y-axis shows the mean cross–track error (CTE) computed over successful runs. Across uncertainty levels, the proposed framework FT and QT triggers sustain high success rates while maintaining comparatively low CTE relative to the baselines.

\begin{figure*}
  \centering
  \setlength{\abovecaptionskip}{0pt}
  \includegraphics[width=\textwidth]{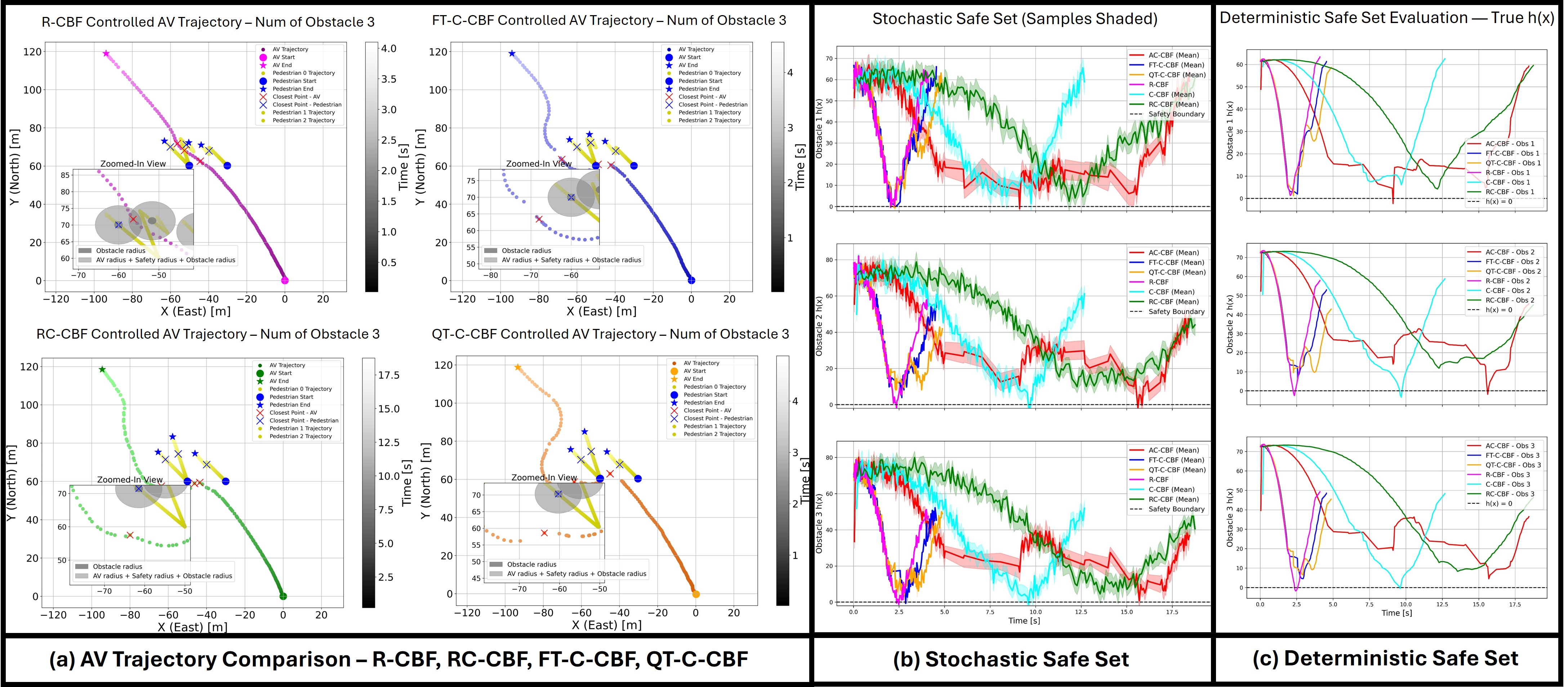}
  \caption{(a) AV trajectory comparison for R-CBF, RC-CBF, FT-C-CBF, and QT-C-CBF under $\sigma_v=0.1, \sigma_o=5$ with three pedestrians.
  (b) Stochastic safe set: sampling-based barrier values $h(x)$ with sample mean (solid) and sampled range (shaded); shading appears when CVaR control is active.
  (c) Deterministic safe set: $h(x)$ evaluated at true states of vehicle and pedestrians without sampling, highlighting clearance relative to $h(x)=0$.}
  \label{one_sce_results}
\end{figure*}

\paragraph*{Most Challenging Scenario}
The table summarizes all pedestrian interaction scenarios with $\sigma_v=0.1$ and $\sigma_o=5\,\mathrm{m}$. Table~\ref{tab:s3_area5} indicates that the conservative baselines C\text{-}CBF and RC\text{-}CBF achieve the highest success rates (SR $=98\%, 97\%$), but maintain very large minimum distances to pedestrians (MDP $=11.3, 12\,\mathrm{m}$), almost twice the required $5.8\,\mathrm{m}$. This conservatism is accompanied by larger cross-track error (CTE $=4.2, 4.3\,\mathrm{m}$) and notable computation time (CT $=56.2, 75.3\,\mathrm{ms}$). 
In contrast, the performance-based R\text{-}CBF has the lowest SR ($67\%$) with the smallest MDP ($6.3\,\mathrm{m}$) and the smallest CTE ($3.9\,\mathrm{m}$), reflecting aggressive behaviors that occasionally fails. Our design goal is to preserve C\text{-}CBF level safety while driving MDP and CTE toward the R\text{-}CBF baseline. 

Both proposed triggers, FT--C--CBF and QT--C--CBF meet that goal: they have 94--96\% success rate without colliding with pedestrian with substantially lower CT (24--26\,ms) and the lowest tracking error (CTE $=3.2$--$3.6\,\mathrm{m}$), while keeping a moderate MDP (7.1--8.3\,m). Their infeasible rates (IR $=15.8\%/16.1\%$) are comparable to the baselines (e.g., C-CBF $=15.2\%$). 
Other variants, AC-CBF (SR $=84\%$), either pay a higher computational cost or reduce task completion. Overall, FT-C-CBF and QT-C-CBF offer the best performance and safety in pedestrian interactions, approaching the performance profile of R-CBF while maintaining near C-CBF success.

Figure~\ref{one_sce_results} shows one of the test runs results, in (a) four control trajectories has been provided: R-CBF collides between the second and third pedestrians; RC-CBF yields the farthest detours; the proposed FT-C-CBF and QT-C-CBF maintain moderate clearance and recover to the reference faster than RC-CBF. Panel (b) plots $h(x)$ with sample mean (solid) and sampled range (shaded). Two effects are visible: (i) FT limits the duration of low--$h$ excursions by switching when feasibility is threatened, while QT prevents deep dips by switching earlier on quality degradation. Panel (c) plots the barrier value \(h(x)\) along the true state of the vehicle and pedestrian trajectory, reflecting the true safe boundary situation. The R-CBF drops below \(h=0\), indicating a collision event. In contrast, C-CBF (and RC-CBF) remain above the boundary safe but conservative. The FT and QT triggers skim the boundary and rapidly recover, consistent with their moderate minimum distance (MDP) and low CTE achieved by switching between R\text{-}CBF and CVaR\text{-}CBF.

\begin{table}
\centering
\caption{The comparison is under $\sigma_v=0.1, \sigma_o=5$ with three pedestrian interactions, with total 300 runs.}
\label{tab:s3_area5}
\begin{tabular}{lccccc}
\toprule
Method & SR (\%) & \makecell{MDP (m)\\($\ge$\,5.8 m)} & IR (\%) & CT (ms) & CTE (m) \\
\midrule
R-CBF & 67\% & 6.3 & 11.9 \% & 22.4 & 3.9 \\
C-CBF & \textbf{98\%} & 11.3 & 15.2\% & 56.2 & 4.2 \\
AC-CBF & 84\% & 7.33 & 17.8\% & 76.9 & 5.0 \\
RC-CBF & 97\% & 12 & 15.3\% & 75.3 & 4.3 \\
\rowcolor{gray!15}
FT-C-CBF & 94\% & \textbf{7.1} & 15.8\% & \textbf{24.2} & \textbf{3.2} \\
\rowcolor{gray!15}
QT-C-CBF & 96\% & 8.3 & 16.1\% & 25.7 & 3.6 \\
\bottomrule
\end{tabular}
\end{table}

\begin{table}
\centering
\caption{300 runs across scenarios; $\delta\in\{0.1,1,2\}$.}
\label{tab:delta_by_ctrl_012}
\begin{tabular}{lcccccc}
\toprule
 & \multicolumn{2}{c}{$\delta=0.1$} & \multicolumn{2}{c}{$\delta=1$} & \multicolumn{2}{c}{$\delta=2$} \\
\cmidrule(lr){2-3} \cmidrule(lr){4-5} \cmidrule(lr){6-7}
 & FT & QT & FT & QT & FT & QT \\
\midrule
SR & 94\% & \textbf{98\%} & 94\% & \textbf{97\%} & 93\% & \textbf{96\%} \\
MDP (m) & \textbf{7.26} & 9.71 & \textbf{7.05} & 9.04 & \textbf{7.2} & 8.8 \\
CTE (m) & \textbf{3.5} & 4.14 & \textbf{3.5} & 4.1 & \textbf{3.5} & 3.9 \\
CVaR rate & 0.92\% & 2.61\% & 0.96\% & 2.70\% & 1.04\% & 2.93\% \\
\bottomrule
\end{tabular}
\end{table}

\subsection{Delta Comparison}

With the risk window fixed at $(W,M)=(5,1)$, we sweep the safety margin parameter $\delta\in\{0.1,1,2\}$ across the three pedestrian scenarios, the risk cap is calculated using \eqref{eq:nu_cap} to be \text{$\overline{\nu}$= 0.38, 3.8, 7.6} corresponding to the $\delta$ in (Table~\ref{tab:delta_by_ctrl_012}). Because $\delta$ is coupled to the admissible risk slack $\bar{\nu}$, varying $\delta$ primarily shifts the threshold that triggers the conservative CVaR controller, without otherwise altering the nominal MPC objective. In addition to SR, MDP, and CTE, we report the CVaR rate, defined as the percentage of steps at which the CVaR controller is engaged. Consistent with design, the activation rate increases monotonically with $\delta$ for both FT and QT switching, reflecting a larger safety margin that induces earlier and more frequent activations. 

\paragraph*{FT vs.\ QT}
From Table~\ref{tab:delta_by_ctrl_012}, QT has a consistently higher SR
than FT at each tested $\delta$. This follows from the QT design: the
risk–budget monitor is enforced at every step, so trajectories that are
optimally feasible but do not meet the desired safety margin. QT therefore triggers the $\mathrm{CVaR}$–CBF more often than FT, which is reflected in the higher $\mathrm{CVaR}$ activation rates reported in Table~\ref{tab:delta_by_ctrl_012}.

\section{Conclusion}

We presented a hybrid control framework with a risk monitor to switch the safety filter between the performance CBF and conservative CVAR in autonomous driving scenarios. Monte–Carlo evaluations with localization and obstacle state uncertainty show that using the QT option in the control framework has a high success rate ($\ge\!95\%$) and maintains low CTE showing the trajectory following performance while avoiding the pedestrians. When safety quality is high, the risk monitor allow the controller in its performance mode; as the margin degrades, the monitor reliably triggers the CVaR controller to preserve safety throughout the trajectory. The proposed framework achieves a performance safety balanced with risk-budgeted guarantees and potential computation compatible with on–board control loops. As next steps, we will deploy the framework under the proposed risk-budgeted monitor on our autonomous vehicle platform for hardware–in–the–loop and on–road validation.

\bibliographystyle{IEEEtran}
\vspace{-0.5em}

\bibliography{references}

@INPROCEEDINGS{ames2014cbfacc,
  author={Ames, Aaron D. and Grizzle, Jessy W. and Tabuada, Paulo},
  booktitle={53rd IEEE Conference on Decision and Control}, 
  title={Control barrier function based quadratic programs with application to adaptive cruise control}, 
  year={2014},
  volume={},
  number={},
  pages={6271-6278},
  doi={10.1109/CDC.2014.7040372}}

@article{ames2017cbfqpsafety,
   title={Control Barrier Function Based Quadratic Programs for Safety Critical Systems},
   volume={62},
   ISSN={1558-2523},
   url={http://dx.doi.org/10.1109/TAC.2016.2638961},
   DOI={10.1109/tac.2016.2638961},
   number={8},
   journal={IEEE Transactions on Automatic Control},
   publisher={Institute of Electrical and Electronics Engineers (IEEE)},
   author={Ames, Aaron D. and Xu, Xiangru and Grizzle, Jessy W. and Tabuada, Paulo},
   year={2017},
   month=aug, pages={3861–3876} }

@INPROCEEDINGS{Allibhoy2021,
  author={Allibhoy, Ahmed and Cortés, Jorge},
  booktitle={2021 60th IEEE Conference on Decision and Control (CDC)}, 
  title={Anytime Solution of Constrained Nonlinear Programs via Control Barrier Functions}, 
  year={2021},
  volume={},
  number={},
  pages={6527-6532},
  doi={10.1109/CDC45484.2021.9683035}}

@article{Wang_2023,
   title={Suboptimal Safety-Critical Control for Continuous Systems Using Prediction–Correction Online Optimization},
   volume={53},
   ISSN={2168-2232},
   url={http://dx.doi.org/10.1109/TSMC.2023.3240290},
   DOI={10.1109/tsmc.2023.3240290},
   number={7},
   journal={IEEE Transactions on Systems, Man, and Cybernetics: Systems},
   publisher={Institute of Electrical and Electronics Engineers (IEEE)},
   author={Wang, Shengbo and Wen, Shiping and Yang, Yin and Cao, Yuting and Shi, Kaibo and Huang, Tingwen},
   year={2023},
   month=jul, pages={4091–4101} }

@article{clark2021cbfstochastic,
title = {Control barrier functions for stochastic systems},
journal = {Automatica},
volume = {130},
pages = {109688},
year = {2021},
issn = {0005-1098},
doi = {https://doi.org/10.1016/j.automatica.2021.109688},
url = {https://www.sciencedirect.com/science/article/pii/S0005109821002089},
author = {Andrew Clark}
}

@INPROCEEDINGS{mrdjan2018inputdelay,
  author={Jankovic, Mrdjan},
  booktitle={2018 Annual American Control Conference (ACC)}, 
  title={Control Barrier Functions for Constrained Control of Linear Systems with Input Delay}, 
  year={2018},
  volume={},
  number={},
  pages={3316-3321},
  doi={10.23919/ACC.2018.8430747}}

@ARTICLE{xiao2022adaptivecbf,
  author={Xiao, Wei and Belta, Calin and Cassandras, Christos G.},
  journal={IEEE Transactions on Automatic Control}, 
  title={Adaptive Control Barrier Functions}, 
  year={2022},
  volume={67},
  number={5},
  pages={2267-2281},
  doi={10.1109/TAC.2021.3074895}}

@misc{rabiee2024closedformcontrolsafetyinput,
      title={A Closed-Form Control for Safety Under Input Constraints Using a Composition of Control Barrier Functions}, 
      author={Pedram Rabiee and Jesse B. Hoagg},
      year={2024},
      eprint={2406.16874},
      archivePrefix={arXiv},
      primaryClass={eess.SY},
      url={https://arxiv.org/abs/2406.16874}, 
}

@inproceedings{Agrawal_2021,
   title={Safe Control Synthesis via Input Constrained Control Barrier Functions},
   url={http://dx.doi.org/10.1109/CDC45484.2021.9682938},
   DOI={10.1109/cdc45484.2021.9682938},
   booktitle={2021 60th IEEE Conference on Decision and Control (CDC)},
   publisher={IEEE},
   author={Agrawal, Devansh R. and Panagou, Dimitra},
   year={2021},
   month=dec, pages={6113–6118} }

@ARTICLE{xiao2022hocbf,
  author={Tan, Xiao and Cortez, Wenceslao Shaw and Dimarogonas, Dimos V.},
  journal={IEEE Transactions on Automatic Control}, 
  title={High-Order Barrier Functions: Robustness, Safety, and Performance-Critical Control}, 
  year={2022},
  volume={67},
  number={6},
  pages={3021-3028},
  doi={10.1109/TAC.2021.3089639}}

@misc{sabouni2025reinforcementlearningbasedrecedinghorizon,
      title={Reinforcement Learning-based Receding Horizon Control using Adaptive Control Barrier Functions for Safety-Critical Systems}, 
      author={Ehsan Sabouni and H. M. Sabbir Ahmad and Vittorio Giammarino and Christos G. Cassandras and Ioannis Ch. Paschalidis and Wenchao Li},
      year={2025},
      eprint={2403.17338},
      archivePrefix={arXiv},
      primaryClass={eess.SY},
      url={https://arxiv.org/abs/2403.17338}, 
}

@ARTICLE{ahmadi2022riskcbf,
  author={Ahmadi, Mohamadreza and Xiong, Xiaobin and Ames, Aaron D.},
  journal={IEEE Control Systems Letters}, 
  title={Risk-Averse Control via CVaR Barrier Functions: Application to Bipedal Robot Locomotion}, 
  year={2022},
  volume={6},
  number={},
  pages={878-883},
  doi={10.1109/LCSYS.2021.3086854}}

@misc{wang2025safenavigationuncertaincrowded,
      title={Safe Navigation in Uncertain Crowded Environments Using Risk Adaptive CVaR Barrier Functions}, 
      author={Xinyi Wang and Taekyung Kim and Bardh Hoxha and Georgios Fainekos and Dimitra Panagou},
      year={2025},
      eprint={2504.06513},
      archivePrefix={arXiv},
      primaryClass={cs.RO},
      url={https://arxiv.org/abs/2504.06513}, 
}

@misc{xiao2023learningfeasibilityconstraintscontrol,
      title={Learning Feasibility Constraints for Control Barrier Functions}, 
      author={Wei Xiao and Christos G. Cassandras and Calin A. Belta},
      year={2023},
      eprint={2303.09403},
      archivePrefix={arXiv},
      primaryClass={math.OC},
      url={https://arxiv.org/abs/2303.09403}, 
}

@INPROCEEDINGS{ames2019controlbarrierfunctions,
  author={Ames, Aaron D. and Coogan, Samuel and Egerstedt, Magnus and Notomista, Gennaro and Sreenath, Koushil and Tabuada, Paulo},
  booktitle={2019 18th European Control Conference (ECC)}, 
  title={Control Barrier Functions: Theory and Applications}, 
  year={2019},
  volume={},
  number={},
  pages={3420-3431},
  doi={10.23919/ECC.2019.8796030}}

@inbook{sergey2008cvar,
author = {Sergey Sarykalin and Gaia Serraino and Stan Uryasev},
title = {Value-at-Risk vs. Conditional Value-at-Risk in Risk Management and Optimization},
booktitle = {State-of-the-Art Decision-Making Tools in the Information-Intensive Age},
chapter = {Chapter 13},
pages = {270-294},
doi = {10.1287/educ.1080.0052},
URL = {https://pubsonline.informs.org/doi/abs/10.1287/educ.1080.0052},
eprint = {https://pubsonline.informs.org/doi/pdf/10.1287/educ.1080.0052}
}

@article{rockafellar2000optimization,
  title={Optimization of conditional value-at-risk},
  author={Rockafellar, R Tyrrell and Uryasev, Stanislav and others},
  journal={Journal of risk},
  volume={2},
  pages={21--42},
  year={2000},
  publisher={Citeseer}
}

@inproceedings{Renganathan2025Experimental,
  author       = {Renganathan, Vishnu and Chang, Pei Yu and Ahmed, Qadeer},
  title        = {Experimental Evaluation of Model Predictive Control and Control Barrier Functions for Autonomous Driving Obstacle Avoidance},
booktitle ={Proceedings of the 13th IFAC Symposium on Nonlinear Control Systems (NOLCOS).},
  year         = {2025},
  organization = {International Federation of Automatic Control (IFAC)}
}

\end{document}